\documentclass[12pt]{article}
\usepackage{amsmath}
\usepackage{graphicx,psfrag,epsf}
\usepackage{enumerate}
\usepackage{natbib}
\usepackage{amsmath,amsthm,amsfonts,epsfig,amssymb,natbib,eucal,eufrak}
\usepackage{booktabs}
\usepackage{multirow}
\usepackage{siunitx}
\usepackage{qtree}
\usepackage{setspace}

\usepackage{float}
\restylefloat{table}
\usepackage{color}

\newtheorem{thm}{Theorem}

\newtheorem{ex}{Example}

\newtheorem{comment}{Comment}

\newtheorem{pro}{Proposition}

\newcommand{\blind}{0}

\begin{document}

\def\spacingset#1{\renewcommand{\baselinestretch}%
{#1}\small\normalsize} \spacingset{1}


\if0\blind
{
  \title{\bf On optimal policy in the group testing with incomplete identification}
  \author{Yaakov Malinovsky
  \\
   Department of Mathematics and Statistics\\ University of Maryland, Baltimore County, Baltimore, MD 21250, USA\\
}
  \maketitle
} \fi

\if1\blind
{
  \bigskip
  \bigskip
  \bigskip
  \begin{center}
    {\LARGE\bf Title}
\end{center}
  \medskip
} \fi

\bigskip
\begin{abstract}
Consider a very large (infinite) population of items, where each item independent from the others is defective with probability
$p$, or good with probability $q=1-p$. The goal is to identify $N$ good items as quickly as possible. The following group testing policy (policy A) is considered: test items together in the groups, if the test outcome of group $i$ of size $n_i$ is negative, then accept all items in this group as good, otherwise discard the group. Then, move to the next group and continue until exact $N$ good items are found. The goal is to find an optimal testing configuration, i.e., group sizes, under policy A, such that the expected waiting time to obtain $N$ good items is minimal.
Recently, \cite{G2012} found an optimal group testing configuration under the assumptions of constant group size and $N=\infty$.
In this note, an optimal solution under policy A for finite $N$ is provided.
\end{abstract}

\noindent%
{\it Keywords:} Dynamic programming; Optimal design; Partition problem; Shur-convexity \vfill

\newpage
\doublespacing
\section{Introduction and problem formulation }
\label{se:I}
Consider a subset of $x$ items, where each item has the probability $p$ to be defective, and $q=1-p$ to be good independently from the other items. Following the accepted notation in the group testing literature, we call that model a binomial model \citep{SG1959}. A group test applied to the subset $x$ is a binary test with two possible outcomes, positive or negative.
The outcome is negative if all $x$ items are good, and the outcome is positive if at least one item among $x$ items is defective.

In $1943,$ Robert Dorfman introduced the concept of group testing based on the need to administer syphilis tests to a very large number of individuals drafted into the U.S. army during  World War II. The goal was {\it complete identification} of all drafted people. The Dorfman procedure \citep{Dorfman1943} is a two-stage procedure, where the group is tested in the first stage and if the outcome is positive, then in the second stage individual testing is performed. If the group test outcome is negative in the first stage, then all items in the group are accepted as good. In this simple procedure, the saving of time may be substantial, especially for the small values of $p$. For example, if
$p=0.01$, when compared with individual testing, the reduction in the expected number of tests is about $80\%$.

Since the Dorfman work, group testing has wide-spread applications from communication networks \citep{W1985} to DNA and blood screening \citep{Dh2006, B2016}. Until today, an optimal group testing procedure for complete identification under binomial model is unknown for $\displaystyle p<(3-5^{1/2})/2$.  For $\displaystyle p\geq(3-5^{1/2})/2$ \cite{U1960} proved that the optimal group testing procedure is an individual, one-by-one testing (at the boundary point it is an optimal). However, substantial improvements of Dorfman's procedure were obtained \citep{S1957, SG1959, H1976}. For the review and comparisons among group testing procedures under binomial model see \cite{MA2016}.

To the best of our knowledge, the {\it incomplete identification} problem was introduced by \cite{BP1990} and extended by \cite{BPP1995}.
In their model, demand $D$ of good items should be fulfilled by purchasing two kinds of items. The first kind is $100\%$ quality items with the purchasing cost $s$ per unit, and the second kind is $100q\%$ quality items with purchasing cost $c$ per unit. In addition, there is cost $K$ for each group-test regardless of the size of the tested group with the items of $100q\%$ quality. Under these constrains/assumptions, the authors found an optimal number of $100q\%$ quality to purchase (once) and an optimal group
size chosen from the purchased group, in each stage of the testing procedure. It is related to the problem we discuss here, but with different assumptions and constrains.

Consider the binomial model with a very large (infinite) population of items. The goal is to identify $N$ good items as quickly as possible.
This is an {\it incomplete identification} problem. We consider the following group testing policy (policy A): Test items together in the groups, if test outcome of the group $i$ of size $n_i$ is negative, then accept all items in this group as good, otherwise discard the group. Then, move to the next group and continue until exact $N$ good items will be found. The goal is to find an optimal testing configuration, i.e., group sizes, under policy A, such that the expected waiting time to obtain $N$ good items is minimal.

In the recent work \citep{G2012} the problem of incomplete identification was considered.
The policy A was applied under assumptions $N=\infty$ and a constant group size. The author found an optimal group size as a function of $q$. It can be explained as follows: Each time a group of size $n$ is tested, if the test outcome of the group is negative, then accept all $n$ items in this group as good, otherwise discard the group and take the next group of size $n$ and so on.
The waiting time (number of tests until first good group) is a geometric random variable with expectation $\displaystyle \frac{1}{q^n}$.
Therefore, the mean waiting time per one good item is $\displaystyle \frac{1}{nq^n}$. We want to minimize this quantity. It is equivalent to maximizing the function $\mu(n,\,q)=nq^n$, which is concave as a function of continuous variable $n$. But, since the feasible solution is an integer, the
maximizer is not necessarily unique. In the proposition below we present a slightly modified result by \cite{G2012}, which found an optimal group size as the function of $q$. We also follow the accepted notation in the group testing literature and denote $p$ as the probability to be defective, which is differ from \cite{G2012} notation.

\begin{pro}[\citep{G2012}]
Define $\displaystyle n^{**}=\frac{1}{\ln(1/q)}$.  Under policy A with the constant group size and $N=\infty$ , the optimal group size
for the $q\geq 1/2$ is
\begin{equation}
\label{eq:Gusev}
n^{*}=\left\{
\begin{array}{ccccc}
  n^{**} & if & integer \\
  \lfloor n^{**}\rfloor & if  & \mu(\lfloor n^{**}\rfloor,\,q) > \mu(\lceil n^{**}\rceil,\,q) \\
  \lceil n^{**}\rceil & if & \mu(\lfloor n^{**}\rfloor,\,q) < \mu(\lceil n^{**}\rceil,\,q)\\
   \lfloor n^{**}\rfloor \,\, or\,\, \lceil n^{**}\rceil& if & \mu(\lfloor n^{**}\rfloor,\,q) = \mu(\lceil n^{**}\rceil,\,q),
\end{array}
\right.
\end{equation}
where $\displaystyle \lfloor x\rfloor \left(\displaystyle \lceil x \rceil\right)$ for $x>0$ is defined as the largest (smallest) integer, which is smaller (larger) or equal to $x$. For $q <1/2$, the optimal group size $n^{*}$ equals $1$.
\end{pro}
\bigskip

\begin{comment}{(Cut-off point)}\\
There is an analogy of Ungar's cut-off point for the complete identification.
It seems that for $N=2$, the policy $A$ with the groups of size 2 is the only reasonable policy for an incomplete identification problem.
For $N=2$, policy $A$ is better than the individual testing if the expected waiting time $\displaystyle 1/q^2$ is less than the expected waiting time
$2/q$ under individual testing, i.e., $q>1/2$.
Now, following \cite{U1960} with adoption to incomplete identification case, one can show that if $q <1/2$, then individual testing is the optimal among all possible strategies for any $N$.
In the boundary case $q=1/2$, the individual testing is an optimal strategy.
\end{comment}

\newpage
\noindent
{\bf The problem formulation: Finite $N$}

Under policy $A$, we are interested in finding an optimal partition $\displaystyle \left\{m_1,\ldots,m_{J}\right\}$ with $\displaystyle m_1+\ldots+m_{J}$  for some $\displaystyle J\in\left\{1,\ldots,N\right\}$ such that the expected total waiting time to obtain $N$ good items is minimal, i.e.,
\begin{align}
\label{eq:Goal}
&
\displaystyle \left\{m_1,\ldots m_J\right\}=\arg\min_{n_1,\ldots,n_I}
\left\{\frac{1}{q^{n_1}}+\ldots+\frac{1}{q^{n_{I}}}\right\},\nonumber\\
&
\text{subject to}\,\,\,\,\,\, \sum_{i=1}^{I}n_i=N,\,\,\, I\in\left\{1,\ldots,N\right\}.
\end{align}

\section{Dynamic Programming algorithm and alternative efficient solutions}

Denote $n\,\left(n=1,\ldots,N\right)$ as a number of good items remains yet unidentified and $H(n)$ an optimal total expected time to obtain $n$
good items.
Then, if we test a group of size $x\,(x=1,\ldots, n)$, we have
\begin{equation}
\label{eq:Inf}
H(n)=q^{x}H(n-x)+(1-q^x)H(n),\,\,n=2,\ldots,N;\,\,\,x=1,\ldots,n,
\end{equation}
where $H(0)=0,\, H(1)=1.$

Combining $H(n)$ from the left and right-hand side of \eqref{eq:Inf} we obtain the dynamic programming (DP) algorithm:

\begin{align}
\label{eq:DP}
&
H(0)=0, H(1)=1,\\\nonumber
&
H(n)=\min_{x=1,\ldots,n}\left\{\frac{1}{q^x}+H(n-x)\right\},\,\,\,\,n=2,\ldots,N.
\end{align}

The complexity of calculation of $H(N)$ is $O(N^2)$.

We present below two examples that help to illustrate how subgroup sizes may differ.
\begin{ex}
$N=250,\,p=0.01$. In this case $n^{*}=99$ or $100$.
An optimal DP algorithm for the problem $A$ gives us the unique (until permutation) solution:\\
$n_1=83,\,n_2=83,\, n_3=84$ with expected waiting time equals to $6.9320$.
\end{ex}

\begin{ex}
$N=220,\,p=0.01$, then the optimal solution is $n_1=n_2=110$ with expected waiting time $6.0417$.
\end{ex}

Both examples provide insight on the possibility that under an optimal policy $A$, subgroup sizes differ at most by
one unit.
In the following proposition we prove this conjecture.
This result also allows us to reduce the computational complexity of $H(N)$ from $O(N^2)$ to $O(N)$.

\begin{pro}
\label{re:Major}
For the given $\displaystyle I,\,n_1,\ldots,n_I$ exist  $\displaystyle n^{*}_1,\ldots,n^{*}_I$ with $\displaystyle |n^{*}_i-n^{*}_j|\leq 1,\,\,\text{for all}\,\,\,i,j=1,\ldots,I$, and $\displaystyle \sum_{i=1}^{I}n_i=\sum_{i=1}^{I}n^{*}_i$ such that
\begin{equation}
\label{eq:E}
\displaystyle
{
\frac{I}{q^{{\overline{n}}_I}} \leq \frac{1}{q^{n^{*}_1}}+\ldots+\frac{1}{q^{n^{*}_{I}}} \leq \frac{1}{q^{n_1}}+\ldots+\frac{1}{q^{n_{I}}},
}
\end{equation}
where $\displaystyle {\overline{n}}_I=\frac{n_1+\ldots+n_I}{I}.$
\end{pro}

\begin{proof}
The function $\displaystyle f\left(x_1,\ldots,x_I\right)=\frac{1}{q^{x_1}}+\ldots+\frac{1}{q^{x_{I}}}$ is Shur-convex on the finite support as the function of continuous variables $x_1,\ldots,x_I$.
Since $\displaystyle{ \left({\overline{n}}_I,\ldots {\overline{n}}_I\right) \prec \left(n^{*}_1,\ldots,n^{*}_I\right) \prec \displaystyle \left(n_1,\ldots,n_I\right)}$, where `$\displaystyle \prec$' denotes majorization (see \cite{S2004}), from Shur-convexity of $\displaystyle f\left(x_1,\ldots,x_I\right)$ we obtain $\displaystyle{ f\left({\overline{n}}_I,\ldots {\overline{n}}_I\right)\leq f\left(n^{*}_1,\ldots,n^{*}_I\right) \leq \displaystyle f\left(n_1,\ldots,n_I\right)}$, and Proposition \ref{re:Major} follows.
\end{proof}

We can use Proposition \ref{re:Major} to solve problem \eqref{eq:Goal} with computational complexity $O(N)$ in the following way.
Fix $I=1,\ldots,N$. If $\displaystyle {\overline{n}}_I=N/I$ is an integer then, due to \eqref{eq:E} it follows that for the fixed number of groups $I$ it is optimal solution, which minimizes the expected total time. Otherwise, any partition with
 $\displaystyle |n^{*}_i-n^{*}_j|\leq 1,\,\,\text{for all}\,\,\,i,j=1,\ldots,I$, and $\displaystyle \sum_{i=1}^{I}n^{*}_i=N$ is an optimal.
As such, we have to repeat the algorithm for $I=1, I=2, \ldots, I=N$, the computational complexity is $O(N)$.

Finally, in the spirit of \cite{G1985}, we present an optimal solution for the optimization problem \eqref{eq:Goal}.
The solution uses the value of $\displaystyle n^{*}$ and allows us to reduce the optimization problem \eqref{eq:Goal}
in such a way that we have to consider only up to the two partitions. Then, we need to evaluate the expected total waiting time for each of these partitions and choose one with the minimal expectation.

\begin{thm}
\label{re:Main}
Suppose $q>1/2$.
In the case of non-unique solution in \eqref{eq:Gusev}, choose $n^{*}=\lfloor n^{**}\rfloor$.
Denote $\displaystyle s=\Big \lfloor\frac{N}{n^{*}}\Big \rfloor$ and $\theta=N-sn^{*}$. Then, an optimal partition
$\displaystyle \left\{m_1,\ldots, m_J\right\}$ under policy $A$, i.e., an optimal solution in the equation \eqref{eq:Goal} is the following:
\begin{enumerate}
\item[(i)]
If\,\, $\theta=0$ then $\displaystyle m_1=\ldots=m_J=n^{*}$,
\item [(ii)]
If\,\, $1\leq \theta\leq s$ and $n^{*}$ in \eqref{eq:Gusev} is not unique, then an optimal partition
is $\displaystyle m_1,\ldots,m_s$ with $m_i=\lfloor n^{**}\rfloor$ or $m_i=\lceil n^{**}\rceil$ for all $i=1,\ldots,s$.
\item[(iii)]
Otherwise, an optimal partition is one of the following:
\begin{itemize}
\item[(a)]
Distribute $\theta$ among $s$ groups (with initial size $n^{*}$) in such a way that $\displaystyle |m_i-m_j|\leq 1$
for all $i,j \in \left\{1,\ldots,s\right\}$.
\item[(b)]
Build up an additional group (group s+1) by taking the reminder $\theta$ and units from the above $s$ groups (with initial size $n^{*}$)
in such way that $\displaystyle |m_i-m_j|\leq 1$
for all $i,j \in \left\{1,\ldots,s+1\right\}$.
\end{itemize}
\end{enumerate}
\end{thm}

\begin{proof}
\begin{itemize}
\item[(i)]
Follows from the convexity of the function $\displaystyle 1/q^{x}$ as a continuous variable $x$ for the fixed $q$.
\item[(ii)]
Again, using the convexity we have
$$\displaystyle \frac{1}{q^{\lfloor n^{**}\rfloor}}+\frac{1}{q^{\lceil n^{**}\rceil}}<\frac{1}{q^{n_1}}+\frac{1}{q^{n_2}},$$
for any $n_1,n_2$ such that $\displaystyle n_1+n_2=\lfloor n^{**}\rfloor+\lceil n^{**}\rceil$ and $\left\{n_1, n_2\right\}\neq\left\{\lfloor n^{**}\rfloor,\lceil n^{**}\rceil\right\}$,
and the result follows.
\item[(iii)]
From Proposition \ref{re:Major} we know that under an optimal policy $A$, subgroup sizes differ at most by
one unit.
Consider case (a), i.e., $s$ groups with $s=u_1+u_2,$ where $u_1$ is the number of subgroups of size $n^{*}+t$, for some $t\geq 1$,
and $u_2$ is the number of subgroups of size $n^{*}+t+1$, such that $\displaystyle u_1(n^{*}+t)+u_2(n^{*}+t+1)=N$.
Suppose that we partition $N$ into a fewer subgroups $s_1$ such that $s_1<s$.  Suppose that $s_1=v_1+v_2,$ where $v_1$ is the number of subgroups of size $n^{*}+j$, for some $j>t\geq 0$,
and $v_2$ is the number of subgroups of size $n^{*}+j+1$, such that $\displaystyle v_1(n^{*}+j)+v_2(n^{*}+j+1)=N$.
Denote, $\displaystyle f(x)\equiv f(x, q)=\frac{1}{xq^x}$. For the fixed $q$ the function $f(x)$ is the convex function of the continuous variable $x$.
Therefore, we have
\begin{align}
\label{eq:In}
&
\frac{f(n^{*}+t)}{n^{*}+t}< \frac{f(n^{*}+t+1)}{n^{*}+t+1}< \frac{f(n^{*}+j)}{n^{*}+j}< \frac{f(n^{*}+j+1)}{n^{*}+j+1}.
\end{align}
From \eqref{eq:In} we get the following inequality
\begin{align*}
&
u_1f(n^{*}+t)+u_2f(n^{*}+t+1)< v_1f(n^{*}+j)+v_2f(n^{*}+j+1).
\end{align*}
Therefore, the partition into fewer than $s$ subgroups cannot be optimal.
Consider case (b): The similar arguments lead to the conclusion that partitioning into more than $s+1$ subgroups cannot be optimal.

\end{itemize}
\end{proof}
\section{Discussion}
In this work, we provide an optimal solution under policy $A$ for an incomplete identification problem.
It is a natural complement to the recent investigation by \cite{G2012}.
There is an interesting open question: Overall, is policy A an optimal policy for the incomplete identification problem in the sampling
from infinite population with finite demand $N$?\\
Also, it is important to note that we assume that parameter $p$ is known.
In many practical situations, the parameter is unknown or the limited knowledge is available, such as the upper bound or lower bound.
In this case, the Bayesian methodology from the complete identification case \citep{SG1966} or the minimax method \citep{MA2015} can be adopted.
Another possible direction for the investigation is to remove the assumption that the tests are error-free.
In this case, the expected total waiting time cannot be used as the only criterion for comparison among group-testing procedures
and additional criteria must be considered \citep{BSS2006, MAR2016}.
We do not attempt to investigate
erroneous tests in the current work and leave this direction for
future investigations.

\section*{Acknowledgement}
The author thanks the editor for his time and advice.
{}

\end{document}